\documentclass[twocolumn]{article}

\usepackage{amsmath}
\usepackage{amssymb}
\usepackage{amsthm}
\usepackage{graphicx}
\usepackage{framed}
\usepackage{url}
\usepackage{listings}
\usepackage{algorithm}
\usepackage{times}
\usepackage{fullpage}

\newtheorem{definition}{Definition}
\newtheorem{lemma}{Lemma}
\newtheorem{problem}{Problem}

\newcommand{\dmin}{\ensuremath{d_{min}}}
\newcommand{\davg}{\ensuremath{d_{avg}}}
\newcommand{\dlb}{\ensuremath{d_{lb}}}
\newcommand{\demd}{\ensuremath{d_{emd}}}

\newcommand{\oi}{\ensuremath{O_i}}
\newcommand{\oj}{\ensuremath{O_j}}

\newcommand{\ti}{\ensuremath{t_i}}
\newcommand{\tj}{\ensuremath{t_j}}

\newcommand{\kth}{\ensuremath{k^\textrm{th}}}

\newcommand{\mindist}{MinDist}
\newcommand{\avgdist}{AvgDist}
\newcommand{\avgdistbuild}{\avgdist-Build}
\newcommand{\avgdistquery}{\avgdist-Query}
\newcommand{\nextestimate}{NextEstimate}
\newcommand{\complete}{Complete}
\newcommand{\emddist}{EMD}

\newcommand{\lone}{\ensuremath{L_1}}

\newcommand{\comment}[1]{}

\newcommand{\figwidth}{0.85\columnwidth}

\title{Finding Top-k Similar Pairs of Objects Annotated with Terms from an Ontology}

\author{
\hspace*{-5mm}
\begin{tabular}{ccc}
	Arnab Bhattacharya & Abhishek Bhowmick & Ambuj K. Singh \\
	\url{arnabb@iitk.ac.in} & \url{bhowmick@iitk.ac.in} & \url{ambuj@cs.ucsb.edu} \\
	{\small Dept. of Computer Science and Engineering,} & {\small Dept. of Computer Science and Engineering,} & {\small Dept. of Computer Science,}\\
	{\small Indian Institute of Technology, Kanpur} & {\small Indian Institute of Technology, Kanpur} & {\small University of California, Santa Barbara}\\
	{\small Kanpur, UP 208016, India.} & {\small Kanpur, UP 208016, India.} & {\small Santa Barbara, CA 93106, USA.}\\
\end{tabular}
}

\date{}

\begin{document}

\maketitle

\begin{abstract}
	With the growing focus on semantic searches and interpretations, an
	increasing number of standardized vocabularies and ontologies are being
	designed and used to describe data. We investigate the querying of objects
	described by a tree-structured ontology. Specifically, we consider the case
	of finding the top-$k$ best pairs of objects that have been annotated with
	terms from such an ontology when the object descriptions are available only
	at runtime.  We consider three distance measures. The first one defines the
	object distance as the minimum pairwise distance between the sets of terms
	describing them, and the second one defines the distance as the average
	pairwise term distance. The third and most useful distance measure---earth
	mover's distance---finds the best way of matching the terms and computes the
	distance corresponding to this best matching. We develop lower bounds that
	can be aggregated progressively and utilize them to speed up the search for
	top-$k$ object pairs when the earth mover's distance is used. For the
	minimum pairwise distance, we devise an algorithm that runs in $O(D + T k
	\log k)$ time, where $D$ is the total information size and $T$ is the total
	number of terms in the ontology. We also develop a novel best-first search
	strategy for the average pairwise distance that utilizes lower bounds
	generated in an ordered manner.  Experiments on real and synthetic datasets
	demonstrate the practicality and scalability of our algorithms.
\end{abstract}


\section{Introduction}
\label{sec:intro}

We are witnessing an unprecedented growth in annotated information. This growth
has been motivated by a need to share information and, more recently, by a need
to search and analyze objects based on their structure and semantics.  Annotated
objects occur in multiple application domains including language
(\url{http://wordnet.princeton.edu/}), biology
(\url{http://www.geneontology.org}), medical documents
(\url{http://www.nlm.nih.gov/mesh/}), web content
(\url{http://www.semanticweb.org/}), etc.  In all these cases, annotations are
derived from a structured vocabulary or ontology.  An ontology uses a number of
different relationships (e.g., is-a, is-part-of) to organize concepts or
hierarchies.

This paper investigates the analysis of large sets of objects that have been
annotated with terms from a common ontology.  The basic problem we consider is
as follows: Given two sets of objects annotated with terms from a common
ontology, how to find the top-$k$ pairs of objects among the two sets that are
most similar.

The above problem statement requires us to formalize the notion of distance
\emph{between two terms} in a given ontology and then to extend this notion to
distance \emph{between two annotated objects}.  The distance between two terms
can be measured by the shortest path distance on the ontology.

There are a number of definitions for distance (or conversely, similarity)
between objects.  Two obvious definitions are based on the minimum pairwise
distance and the average pairwise distance between the annotations.  The third
one is the earth mover's distance~\cite{ab05} that takes into account the
relative positions of the terms that describes the objects.  We investigate
querying based on these three distance measures.

In this paper, we consider that the object descriptions are submitted in an
online fashion, i.e., they are available \emph{only} at run-time.
As such, \emph{no} pre-processing or index construction or any other
offline processing can be used, and all the computation costs are paid at
run-time.  Even if the distance function used is a metric, the online nature of
the problem renders the use of index structures like the M-tree~\cite{ss21}
infeasible due to their high index construction times.  In a way, this problem
is reminiscent of the computation of spatial joins on objects embedded in the
Euclidean space: the spatial datasets are delivered online and we need to
compute the best spatial matches~\cite{tr08}.  Only that, in the case of
ontologies, the primitive distance is not Euclidean, but computed on a tree.  

The problem we consider here can be extended easily to the case when objects are
annotated with multiple independent ontologies.  We can compute the per-ontology
distance and combine them using an aggregate ranking technique such as the
threshold algorithm~\cite{ss36}.  The problem of finding objects similar to a
given query object (i.e., the $k$-NN problem) reduces to the special case of a
join of the database with a singleton set, the query object.  Similarly, range
queries can be solved by choosing only those pairs having a distance less than
the query range.  While these and other kinds of queries can also be considered
in our setting, the problem of top-$k$ joins exposes the computational and data
management complexities of this domain well, making it the right problem to
consider.

Formally, our problem can be stated as:
\begin{problem}
	\label{prob:problem}
	Given a set of objects each of which is defined by a set of terms from an
	ontology and a distance function $d(\oi, \oj)$ between two objects \oi{} and
	\oj{}, find $k$ pairs of objects $P$ such that for any $(\oi, \oj) \in P$
	and $(O_g, O_h) \notin P$, $d(\oi, \oj) \leq d(O_g, O_h)$.
\end{problem}

\begin{figure}[t]
	\begin{center}
		\includegraphics[width=0.75\columnwidth]{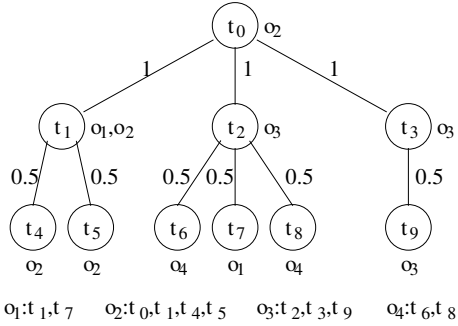}
		\vspace*{-2mm}
		\caption{Example of an ontology tree with objects.}
		\label{fig:example}
	\end{center}
	\vspace*{-6mm}
\end{figure}

Figure~\ref{fig:example} illustrates a particular instance of the problem.  The
ontology tree consists of $10$ terms.  There are $4$ objects that are described
by these terms.  The object descriptions are given by $O_1 = \{t_1, t_7\}$, $O_2
= \{t_0, t_1, t_4, t_5\}$, $O_3 = \{t_2, t_3, t_9\}$, $O_4 = \{t_6, t_8\}$.  An
inverted index, i.e., mapping a term to set of objects can be maintained on the
ontology itself (as shown in the figure).  Thus, each node in the tree
statically maintains a list $L$ of the objects that are described using the term
corresponding to the node.  For example, the list of objects for $t_0$ is
($O_2$).  We will use \emph{term} and \emph{node} interchangeably to denote the
node in which the term resides.

The edge weights on the tree decrease exponentially as the level increases.
Concepts closer to the root of the ontology are less similar than concepts that
share some common ancestors.  For example, broader concepts such as ``sports''
and ``politics'' should be more dissimilar than relatively narrower concepts
such as ``football'' and ``cricket''.  The exponentially decreasing edge weights
capture this notion.  We highlight the fact that the exponential edge weighting
function is an example, and not a necessity for the algorithms to work.  They
produce correct answers for all edge weights.

We denote the number of objects by $N$, the number of terms by $T$, the total
information size (i.e., the total number of describing terms for all the
objects) by $D$, and the number of object pairs queried by $k$.  In
Figure~\ref{fig:example}, $N = 4$, $T = 10$, and $D = 11$.  

Our contributions in this paper are as follows:
\begin{enumerate}
	\item First, we propose the problem of finding top-$k$ most similar object
		pairs that are annotated with terms in a \emph{hierarchy} in an online
		fashion.  The terms may define concepts in an ontology and objects may
		be described using the concepts.
	\item Then, we define and motivate three different distance functions
		(equivalently, similarity measures) that can be used to describe the
		similarity between a pair of objects.  The \emph{minimum pairwise
		distance} is useful for searching objects sharing a similar term
		(concept).  The \emph{average pairwise distance} can be used to query
		objects that are described using multiple attributes.  The \emph{earth
		mover's distance} (EMD) finds the best way of matching the terms from
		two objects and finds the distance corresponding to this best matching.
	\item Finally, we develop efficient algorithms to solve the problem using
		the above distances.  We use lower bounds based on \lone{} on reduced
		number of terms to speed up the computation of EMD.  The \lone{}
		distance, in turn, is computed progressively using a modified version of
		the threshold algorithm.  For the minimum pairwise distance, we show
		that the top-$k$ query runs in $O(T k \log k)$ time, where $T$ is the
		size of the ontology.  For the average pairwise distance, we devise an
		efficient best-first search algorithm that avoids distance computations
		by generating lower bounds in an ordered manner.  Experimental
		evaluations demonstrate the scalability and practicality of our
		algorithms.
\end{enumerate}

The rest of the paper is organized as follows.  Section~\ref{sec:work} describes
the related work.  Section~\ref{sec:distance} defines the term distance and the
different object distances.  Sections~\ref{sec:emd},~\ref{sec:min}
and~\ref{sec:avg} present the different algorithms for finding the top-$k$ pairs
of objects using those distances.  Experimental results are discussed in
Section~\ref{sec:expts}.  Section~\ref{sec:concl} concludes the paper.

\section{Related Work}
\label{sec:work}

Heterogeneous and high-throughput data is becoming commonplace in the sciences
and there is consensus that integration of this information is needed for new
breakthroughs.  In all these cases, annotations are derived from a structured
vocabulary or ontology.  The Semantic Web (\url{http://www.semanticweb.org/})
has defined a specific language, OWL (\url{http://www.w3.org/2004/OWL/}), for
describing ontologies.  In biology, genes are described using Gene Ontology (GO)
(\url{http://www.geneontology.org/}) that annotates genes and gene products by
three kinds of terms reflecting molecular functions, biological processes, and
cellular components.  Millions of abstracts in Pubmed
(\url{http://www.pubmed.gov/}) are indexed using MESH terms
(\url{http://www.nlm.nih.gov/mesh/}).  WordNet
(\url{http://wordnet.princeton.edu/}) is a lexical database that groups English
words into cognitive synonyms (or \emph{synsets}).  Hundreds of other ontologies
have been proposed over diverse application domains such as plant structures
(\url{http://www.plantontology.org/}), description and publication of digital
documents (\url{http://www.dublincore.org/}), and earth and the environment
(\url{http://sweet.jpl.nasa.gov/ontology/}).  A good compendium of different
ontologies is maintained at \url{http://www.ontologyonline.org/}.

A given ontology uses a number of different relationships to organize concepts
or hierarchical relationships.  Of these, ``is-a'' and ``is-part-of''
relationships are the most prevalent.  The former describes a subsumption
relationship while the latter represents how objects combine together to form
composite objects.  Both of these lead to hierarchical structures in which the
proximity between terms (concepts) grows as we descend down the hierarchy.

There have been numerous works on gene ontology ranging from gene function
prediction using information theory~\cite{tr13} to defining similarity among
genes using the full graph structure of GO~\cite{tr10}.  In~\cite{tr11}, a
comparison of three different gene similarity measures were presented.
Probabilistic approaches have also been used~\cite{tr12}.  Biologists have used
average and minimum pairwise distances between genes based on GO for comparing
co-evolutionary rates of yeast genes~\cite{tr15} and for co-clustering with gene
expression data~\cite{tr17} respectively.

There are a number of similar efforts in the area of information retrieval where
the similarity between documents is measured by considering the overlap of
terms.  The term-frequency inverse-document-frequency (tf-idf) measures consider
the frequency of terms in documents~\cite{th30}.  Work on text matching showed
that hierarchy-based measures using tf-idf outperform lexical similarity
measures~\cite{tr14}.  Latent Semantic Indexing (LSI)~\cite{th14} transforms
documents into an Euclidean space indexed by latent semantic dimensions.  EMD
has been shown to be better than other measures in finding document similarities
using the WordNet ontology~\cite{tr18}.

Embedding an ontology into an Euclidean space~\cite{tr09} and processing queries
in the embedding space is another alternative. However, an object description
will then span multiple points leading to possibly large MBRs.  Further, the
approach may suffer from high distortion of the embedding.

In this paper, we tackle the computational challenge of answering queries
efficiently using distances defined on hierarchical structures like ontologies.

\section{Distance Definitions}
\label{sec:distance}

\subsection{Distance between Terms}
\label{sec:termdist}

The distance $d(\ti, \tj)$ between two terms \ti{} and \tj{} is defined as the
length of the path between them on the ontology tree.  Since there is only one
path between two terms in a tree, from the properties of the shortest path, 
this distance is a metric~\cite{tr01}.  

An interesting and important point to note is that when the term distances
decrease \emph{exponentially} at each level, the distance between two terms at
the leaves of two subtrees can be approximated by the distance between the roots
of the subtrees.  For example, in Figure~\ref{fig:example} where the edge
distances are halved at each level, the distance between $t_4$ and $t_6$ (= $3$)
can be approximated by that between $t_1$ and $t_2$ (= $2$).  Using this term
distance, we next define different distance measures between the objects.  Once
more, we emphasize the fact that our algorithms are general enough to work
correctly with all edge weights, and not just the exponential function.

\begin{table}[t]
	\begin{center}
		\begin{tabular}{|c|cccc|}
			\hline
			Min & $O_1$ & $O_2$ & $O_3$ & $O_4$ \\
			\hline
			$O_1$ & 0.0 & 0.0 & 0.5 & 1.0 \\
			$O_2$ & & 0.0 & 1.0 & 1.5 \\
			$O_3$ & & & 0.0 & 0.5 \\
			$O_4$ & & & & 0.0 \\
			\hline
		\end{tabular}
		\caption{Minimum pairwise distances for the example in
		Figure~\ref{fig:example}.}
		\label{tab:minex}
	\end{center}
	\vspace*{-7mm}
\end{table}

\begin{table}[t]
	\begin{center}
		\begin{tabular}{|c|cccc|}
			\hline
			Avg & $O_1$ & $O_2$ & $O_3$ & $O_4$ \\
			\hline
			$O_1$ & 1.25 & 1.50 & 2.08 & 1.75 \\
			$O_2$ & & 0.75 & 2.17 & 2.50 \\
			$O_3$ & & & 1.11 & 2.00 \\
			$O_4$ & & & & 0.50 \\
			\hline
		\end{tabular}
		\caption{Average pairwise distances for the example in
		Figure~\ref{fig:example}.}
		\label{tab:avgex}
	\end{center}
	\vspace*{-7mm}
\end{table}

We next define the three distance measures---$\dmin$, $\davg$ and
$\demd$---between two objects.\footnote{We use the terms $\dmin$ and \mindist{},
$\davg$ and \avgdist{}, $\demd$ and \emddist{} interchangeably in the paper.}

\subsection{Minimum Pairwise Distance}
\label{sec:mindist}

\begin{definition}[Minimum Pairwise Distance]
	\label{def:min}
	The minimum pairwise distance between two objects \oi{} and \oj{}, denoted
	by \emph{\mindist{}}, is defined as:
	\begin{equation}
		\label{eq:min}
		d_{min}(\oi, \oj) = \min_{\ti \in \oi, \tj \in \oj} \left\{	d(\ti, \tj)	\right\}
	\end{equation}
\end{definition}

This distance is useful when searching objects that have similar terms.  For
example, even though a single biological document may contain references to
different terms like photoreceptor cells and ganglion cells, it is useful to be
able to retrieve it when another document that describes photoreceptor cells is
queried.

This distance is of particular use in \emph{keyword searching}, where the query
document consists of only the single keyword, and all documents having that
keyword will be returned with a distance of 0.  \mindist{}, in general, extends
this idea by finding additional documents that contain terms most similar to the
queried keyword.

The \mindist{} measure is heavily used in hierarchical bottom-up clustering
methods where in each step, two clusters with the minimum pairwise distance are
merged.  It has also been successfully used for finding the distance between two
genes, where a gene is annotated with a set of terms from GO~\cite{tr17}.

Table~\ref{tab:minex} shows the \mindist{} measures among the objects in
Figure~\ref{fig:example}.  \mindist{} is \emph{not} a metric distance as it does
not maintain the triangular inequality.  For example, $\dmin(O_1, O_4) +
\dmin(O_1, O_2) = 1.0 + 0.0 < 1.5 = \dmin(O_2, O_4)$.

\subsection{Average Pairwise Distance}
\label{sec:avgdist}

\begin{definition}[Average Pairwise Distance]
	\label{def:avg}
	The average pairwise distance between two objects \oi{} and \oj{}, denoted
	by \emph{\avgdist{}}, is defined as:
	\begin{equation}
		\label{eq:avg}
		d_{avg}(\oi, \oj) = \frac{1}{|\oi|.|\oj|} \sum_{\ti \in \oi, \tj \in
		\oj} d(\ti, \tj) 
	\end{equation}
	where $|\oi|$ and $|\oj|$ denote the number of terms describing
	\oi{} and \oj{} respectively.
\end{definition}

The \avgdist{} is useful in cases where the objects are not precisely defined.
For example, it has been successfully used for gene function prediction using GO
terms for yeast genes~\cite{tr15} as well as in the domain of web
services~\cite{tr16}.  The \mindist{} measure fails in such cases.

Table~\ref{tab:avgex} shows the \avgdist{} measures among the objects in
Figure~\ref{fig:example}.  \avgdist{} is not a metric, as it fails to satisfy
the identity property, i.e., $\davg(x,x)$ can be greater than $0$ (e.g.,
$\davg(O_1,O_1)=1.25$).  However, since it follows symmetry and triangular
inequality\footnote{See Appendix for the proof.}, it can be considered as a
\emph{pseudo-metric} distance.

\subsection{Earth Mover's Distance}
\label{sec:emddist}

\begin{table}[t]
	\begin{center}
		\begin{tabular}{|c|cccc|}
			\hline
			EMD & $O_1$ & $O_2$ & $O_3$ & $O_4$ \\
			\hline
			$O_1$ & - & 1.25 & 1.75 & 1.75 \\
			$O_2$ & & - & 2.17 & 2.50 \\
			$O_3$ & & & - & 2.00 \\
			$O_4$ & & & & - \\
			\hline
		\end{tabular}
		\caption{EMDs for the example in Figure~\ref{fig:example}.}
		\label{tab:emd}
	\end{center}
	\vspace*{-7mm}
\end{table}

Apart from the property of not being a true metric, \avgdist{} also suffers from
the fact that each term in one object is matched with every other term in the
other object.  For example, consider two documents with the terms \{war,
sports\} and \{war, football\}.  Even though it is obvious that the distance
between these two documents should be small, the average distance unnecessarily
compares ``war'' in the first document with ``football'' in the other.  The
earth mover's distance (EMD)~\cite{ab05} rectifies this shortcoming by comparing
only the like terms through finding the best matching between the terms of the
two documents.  For this example, \emddist{} will match ``war'' with ``war'' and
``sports'' with ``football'' and aggregate these distances only.  \emddist{} has
been shown to be better than other distances in finding similar documents using
the WordNet ontology~\cite{tr18}.

Formally, each object is considered to be composed of ``mass'' at the specific
spatial locations (corresponding to the terms that describe the object) in the
ontology.  The total mass of each object is $1$; consequently, the mass at each
term location is inverse of the number of terms describing the object.  For
example, $O_1$ in Figure~\ref{fig:example} will have mass $\frac{1}{2}$
corresponding to terms $t_1$ and $t_7$.

The \emddist{} between two objects $A$ and $B$ is the \emph{minimum} work
required to transform $A$ to $B$, where one unit of work is equal to moving one
unit of mass through one unit of distance in the ontology.  Finding the best
``flows'' (i.e., how much mass needs to be moved from one term in $A$ to another
term in $B$) is a linear programming (LP) problem.
\begin{definition}[Earth Mover's Distance]
	\label{def:emd}
	The earth mover's distance between two objects \oi{} and \oj{}, denoted
	by \emph{\emddist{}}, is defined as:
	\begin{align}
		\label{eq:emd}
		d_{emd}(\oi, \oj) &= \min_f \sum_{t_p \in \oi} \sum_{t_q \in \oj} c_{pq} f_{pq} \\
		\text{s.t., each }& f_{pq} \geq 0, \nonumber \\
		\forall_{t_p \in \oi}, \sum_{t_q \in \oj} f_{pq} = O_{i_p}, &\text{ and } \forall_{t_q \in \oj}, \sum_{t_p \in \oi} f_{pq} = O_{j_q} \nonumber
	\end{align}
	where $c_{pq}$ is the \emph{ground distance} between the terms $t_p$ and
	$t_q$ as per the ontology tree and $O_{i_p}$ is the mass of $t_p$ in \oi{}.
\end{definition}
EMD is a metric when the ground distance is a metric (proof in~\cite{ab05}).
Table~\ref{tab:emd} shows the EMDs among the objects in Figure~\ref{fig:example}.

\begin{figure}[t]
	\begin{center}
		\includegraphics[width=\columnwidth]{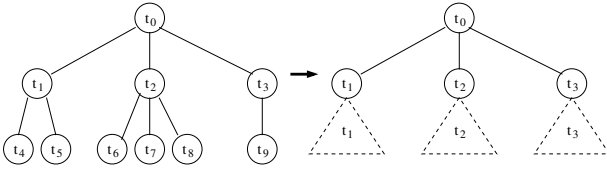}
		\vspace*{-2mm}
		\caption{Reduction of terms.}
		\label{fig:reduction}
	\end{center}
	\vspace*{-6mm}
\end{figure}

\subsection{Comparison of the Distance Measures}
\label{sec:quality}

To compare the usefulness of the three distance measures, we performed the
following experiment.  We used WordNet (\url{http://wordnet.princeton.edu/})
as the ontology and the ``bag-of-words'' dataset from the UCI repository
(\url{http://archive.ics.uci.edu/ml/datasets/Bag+of+Words}) as the set of
objects.  We chose the first $59$ documents from the categories \emph{enron} and
\emph{kos} of the bag-of-words dataset.  Each document was described using nouns
from the WordNet ontology, and the ontology was converted into a tree.  The
top-$50$ pairs were obtained using all the three distances.  For \emddist{}, on
an average, there were $45$ pairs where both the objects were from the same
category.  Also, all $10$ out of the top-$10$ pairs were of this nature.  The
corresponding numbers for \avgdist{} distance were $23$ and $6$ respectively.
The \mindist{} returned $505$ object pairs with distance $0$ as many objects
shared one or more terms.  Consequently, the top-$k$ lists returned were
arbitrary.  This convinced us of the quality of the \emddist{} and its
usefulness in finding the top-$k$ similar pairs of objects described by terms on
tree ontologies.  Nevertheless, the two other distance measures have been proved
to be useful in specific contexts~\cite{tr17,tr15}.

We next design algorithms to efficiently compute the top-$k$ pairs using these
distances.  We start with the \emddist{} as it is the most interesting and
useful measure.

\section{The Algorithm for \emddist{}}
\label{sec:emd}

\begin{table}[t]
	\begin{center}
		\begin{tabular}{|c|c|c|}
			\hline
			Object & Before & After\\
			& $\{t_0, t_1, t_2, t_3, t_4, t_5, t_6, t_7, t_8, t_9\}$ & $\{t_0, t_1, t_2, t_3\}$\\
			\hline
			$O_1$ & $\{0, {1 \over 2}, 0, 0, 0, 0, 0, {1 \over 2}, 0, 0\}$ & $\{0, {1 \over 2}, {1 \over 2}, 0\}$\\
			$O_2$ & $\{{1 \over 4}, {1 \over 4}, 0, 0, {1 \over 4}, {1 \over 4}, 0, 0, 0, 0\}$ & $\{{1 \over 4}, {3 \over 4}, 0, 0\}$\\
			$O_3$ & $\{0, 0, {1 \over 3}, {1 \over 3}, 0, 0, 0, 0, 0, {1 \over 3}\}$ & $\{0, 0, {1 \over 3}, {2 \over 3}\}$\\
			$O_4$ & $\{0, 0, 0, 0, 0, 0, {1 \over 2}, 0, {1 \over 2}, 0\}$ & $\{0, 0, 1, 0\}$\\
			\hline
		\end{tabular}
		\caption{Reduction of terms using Figure~\ref{fig:reduction} for example in
		Figure~\ref{fig:example}.}
		\label{tab:reduction}
	\end{center}
	\vspace*{-7mm}
\end{table}

When the two sets contain $N$ objects each, the problem of finding top-$k$ pairs
of objects can be solved by performing $O(N^2)$ \emddist{} computations.
However, the prohibitive time required by each \emddist{} computation makes the
entire running time ($O(N^2) \times O(EMD) + O(N^2 \log N)$)
impractical.\footnote{The sorting of $N^2$ pairs require an additional $O(N^2
\log N)$ time.}

\subsection{Lower Bound using Reduced Number of Terms}

When the ontology tree has a size of $T$, the ground distance matrix is of size
$T^2$.  However, we need not consider all the terms as we can prune the terms
that are absent in either of the object descriptions.\footnote{The row and
column sums for these terms in the flow matrix will be $0$ and hence all the
flows will be $0$ individually as well.}  Thus, the number of flow variables is
quadratic in the size of the object descriptions.  This is still impractical:
the average time taken to compute $\demd$ for objects of size $7$ was found to
be 54\,ms.\footnote{All the times reported in the paper are based on a 3\,GHz
machine with 2\,GB of RAM running Fedora Linux 9.}

Since the complexity of \emddist{} depends mainly on the number of flow
variables, which is quadratic in the number of terms by which each object is
described, the running time can be reduced if the size of the object
descriptions is reduced.  Figure~\ref{fig:reduction} shows how such reduction can
be accomplished.  The ontology tree is pruned at height $1$; only the root term
and its immediate children remain.  When a term thus deleted appears in an
object description, it is replaced by its ancestor that is retained.  Hence, all
the terms in the dashed subtrees in the figure are removed and replaced by the
root of the subtrees.  The size of an object description is now upper bounded by
the branching factor of the root.  Table~\ref{tab:reduction} shows the reduced
object descriptions.

The \emddist{} between two objects calculated using the reduced ontology is a
lower bound of the \emddist{} using the full ontology~\cite{ab05}\footnote{A
lower bound can be obtained by pruning the tree at any height.  However, there
is a trade-off between the tightness and computational efficiency of the lower
bound.}.  The number of terms in the reduced ontology is generally much less
(say $t \ll T$), thus reducing the number of flow variables to $t^2$.  Since the
complexity of linear programming is at least super-linear in the number of flow
variables, the running time of \emddist{} decreases by a large factor of
$T^2/t^2$.  The number of distance computations, however, still remains
$O(N^2)$.  Next, we show how to reduce the number of distance computations.

\subsection{L$_1$ Lower Bound}

The \lone{} distance, when scaled by the sum of the total mass, can be used as a
lower bound for \emddist{}~\cite{th29}.  Hence, the \lone{} distance between two
objects computed using all $T$ terms, when divided by $2$, serves as a lower bound
for \emddist{} between the objects.  From now on, whenever we mention \lone{},
we mean the scaled version of it which is a lower bound.  \lone{} on all terms,
in turn, is lower bounded by \lone{} on reduced number of terms.  The proof uses
the fact that $|a_i - b_i| + |a_j - b_j| \geq |(a_i + a_j) - (b_i + b_j)|$,
i.e., when the values are combined, the difference of the sums is more than the
sum of the differences.  Therefore, $L_{1_t}(\oi,\oj) \leq L_{1_T}(\oi,\oj) \leq
EMD_T(\oi,\oj)$, where the subscripts denote the number of terms used.  Since
\lone{} is much faster to compute (for $3$ terms, it takes only 0.002\,ms),
we can calculate a lower bound on \emddist{} for each object pair and then use
it as a filtering step to prune many of the pairs.

\begin{figure}
	\begin{center}
		\includegraphics[width=\columnwidth]{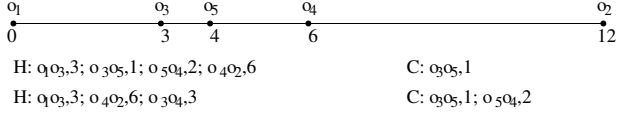}
		\vspace*{-2mm}
		\caption{Example of ordered generation of object pairs for one dimension.}
		\label{fig:sorted}
	\end{center}
	\vspace*{-6mm}
\end{figure}

The \lone{} distance between two objects is a sum of the distances between the
corresponding values in each dimension; therefore, if the distances for all the
object pairs are obtained and sorted for each dimension, the threshold algorithm
(TA)~\cite{ss36} can be applied to obtain the object pairs with the least sum of
distances or the least \lone{} distance in a progressive manner.  The order of
obtaining increasing \lone{} distances can then be used as a guide to order the
\emddist{} computations of the object pairs.

Obtaining a sorted list of object pairs for each dimension requires $O(N^2 \log
N)$ time.  TA, however, also works when the next object pair in the list can be
output in a sorted manner whenever needed.  This avoids the $O(N^2)$
computations.  Hence, now our problem is reduced to outputting the next smallest
pairwise distance whenever asked for in a particular list (or dimension).

For this, we maintain two data structures for each dimension: (i)~a min-heap $H$
that outputs the next best pair, and (ii)~a list $C$ that stores all the pairs
that have been outputted from $H$.

Initially, the $N$ objects are sorted and all $N-1$ consecutive object pairs
(not necessarily $O_iO_{i+1}$) corresponding to $N-1$ differences are inserted
into $H$.  Figure~\ref{fig:sorted} shows an example.  The 5 objects are sorted
according to their values for the dimension that is being processed.  Initially,
$H$ contains the 4 object pairs corresponding to the 4 differences in the sorted
list.  Whenever the next pair is asked for by TA, the minimum object pair from
$H$ is extracted and returned.  It is also inserted into $C$.  In this example,
after the first call, $O_3O_5$ is extracted from $H$ and inserted into $C$.
Similarly, in the next call, $O_5O_4$ is extracted.  

The initial pairs are not sufficient though.  There may be a non-initial pair
(e.g., $O_3O_4$) with a value ($3$) less than that of an initial pair ($O_4O_2$
with value $6$).  However, the important point to note is that any non-initial
pair is a \emph{combination} of some of the initial pairs.  Two pairs which have
an overlapping object can be fused together to generate a new pair.  For
example, $O_3O_4$ can be generated from $O_3O_5$ and $O_5O_4$ since $O_5$ is
overlapping.  Further, a pair can never be the least pair until and unless the
pairs from which it has been generated have been chosen (i.e., output from $H$).
Therefore, in the example, $O_3O_4$ is added to $H$ only after both $O_3O_5$ and
$O_5O_4$ have been chosen.  In general, when a pair $O_xO_y$ is chosen, the
contents of $C$ are scanned and new pairs are generated if possible.  If $C$ has
pairs of the form $O_wO_x$ and $O_yO_z$, new pairs $O_wO_y$ and $O_xO_z$ are
generated respectively and are inserted into $H$.  The value of this new pair is
the sum of the values of the pairs from which it is generated.

\subsection{Algorithm}

Figure~\ref{fig:algoemd} summarizes the entire \emddist{} algorithm that uses TA
with the lower bounding strategy.  First, the \lone{} lower bound using reduced
number of terms is extracted from the heap (line 14).  If it is less than the
current \kth{} estimate $P.dist$ in $P$ (line 15), the bound is improved by
computing the \lone{} using all the terms (line 16).  If it still less (line
17), the exact \emddist{} is computed (line 18) and the top-$k$ list is
modified, if necessary (line 21).  For each such \lone{}-reduced computation,
the threshold distance ($R$) is increased.  When $R >$ \kth{} distance in $P$,
no other object pair can have \lone{}-reduced distance less than the top-$k$
pairs already found.  Therefore, the \emddist{} distances will also be greater.
Hence, the algorithm is then halted.

\begin{figure}[t]
	\begin{center}
		\begin{framed}
			\begin{tabbing}
				Algorithm \textbf{EMD} \\
				\textbf{Input:} Reduced object list $O$ with $t$ terms \\
				\textbf{Output:} Object pair list $P$ \\
				1. \textbf{for} \= dimension $j = 1$ to $t$ \\
				2. \> Sort $O_{ij}$ (only $j^\text{th}$ dimension) \\
				3. \> Insert $N-1$ differences into heap $H[j]$ \\
				4. \> List $C[j]$ := $\Phi$ \\
				5. \> Thresholds $\tau[j]$ := $0$ \\
				6. \textbf{end for} \\
				7. $P$ := $\Phi$ ($\therefore$ $P.dist$ (i.e., $k^\text{th}$ distance in $P$) := $\infty$) \\
				8. Threshold $R$ := sum of all $\tau[j]$ (therefore, $R$ := $0$) \\
				9. $j$ := $0$ \\
				10. \textbf{while} \= $R < P.dist$ \\
				11. \> Extract minimum pair $p$ from $H[j]$ \\
				12. \> $\tau[j]$ := difference for pair $p$ \\
				13. \> \textbf{if} \= $p$ is not seen earlier \\
				14. \> \> $d_1$ := \lone{} on $t$ terms of $p$ \\
				15. \> \> \textbf{if} \= $d_1 \leq P.dist$ \\
				16. \> \> \> $d_2$ := \lone{} on all terms of $p$ \\
				17. \> \> \> \textbf{if} \= $d_2 \leq P.dist$ \\
				18. \> \> \> \> $d_3$ := EMD on all terms of $p$ \\
				19. \> \> \> \> \textbf{if} \= $d_3 \leq P.dist$ \\
				20. \> \> \> \> \> Insert $p$ into $P$ \\
				21. \> \> \> \> \> Update $P.dist$ as new $k^\text{th}$ distance \\
				22. \> \> \> \> \textbf{end if} \\
				23. \> \> \> \textbf{end if} \\
				24. \> \> \textbf{end if} \\
				25. \> \textbf{end if} \\
				26. \> Update $R$ using $\tau[j]$ \\
				27. \> Scan $C[j]$ with $p$ to generate new pairs $\Gamma$ \\
				28. \> Add $\Gamma$ to $H[j]$ \\
				29. \> $j$ := $(j + 1) \mod t$ \\
				30. \textbf{end while}
			\end{tabbing}
		\end{framed}
		\vspace*{-4mm}
		\caption{The \emddist{} algorithm.}
		\label{fig:algoemd}
	\end{center}
	\vspace*{-6mm}
\end{figure}

\subsection{Analysis of Time Complexity}

For each of the $t$ dimensions, we incur the following cost.  Initially, sorting
takes $O(N \log N)$ time.\footnote{An alternative approach using hashing that
may reduce this time is discussed in Appendix.}  Thereafter, inserting the $N-1$
elements in the heap takes $O(N)$ time.  With each call to the heap, an extract
operation takes $O(\log N)$.  At the $i^\text{th}$ iteration, at most $i$
elements are added to the heap again.  This takes $O(i \log N)$ time.  Thus, if
we have $k'$ calls (this $k'$ is generally larger than $k$ as many object pairs
with low lower bound but high \emddist{} are examined), the time per dimension
is $O(k'^2 \log N)$ which leads to a total time of $\sum_{j=1}^t O(k'^2_j \log
N)$ or $O(t k'^2_{max} \log N)$ where $k_{max}$ is the maximum of all $k'_j$s.
\\

\noindent
\textbf{Space Complexity:}
Heap $j$ requires $O(N+k'^2_j)$ space where $k'_j$ is the number of calls made
on column $j$.  Hence, the total space required is $O(t(N+k'^2_{max}))$ where
$k'_{max}$ is the maximum of $k'_j$'s.

\section{The Algorithm for \mindist{}}
\label{sec:min}

Unlike the \demd{} distance, whenever two terms corresponding to two objects are
encountered, the \mindist{} for the object pair can be estimated.  If it is
better than the current estimate, it is retained; otherwise, it is \emph{never}
needed again.  We next explain the \mindist{} algorithm that exploits this property.

\begin{figure}[t]
	\begin{center}
		\begin{framed}
			\begin{tabbing}
				Algorithm \textbf{MinDist} \\
				\textbf{Input:} Node $t$ \\
				\textbf{Output:} \= Pair list $P$ of size $k$; \\
				   \> Object list $B$ of size $O(\sqrt{k})$ \\
				1. $TB$ := list of objects in $t$ of size $O(\sqrt{k})$ \\
				2. $c$ := number of children of $t$ \\
				3. \textbf{for} \= $i$ = $1$ to $c$ \\
				4. \> $CB[i], CP[i]$ := MinDist($t.child[i]$) \\
				5. \> Add $t.edge[i]$ cost to each object in $CB[i]$ \\
				6. \textbf{end for} \\
				7. $B$ := Merge($CB[1]$, $\dots$, $CB[c]$, $TB$) \\
				8. $TP$ := GenPairs($B$) \\
				9. $P$ := Merge($CP[1]$, $\dots$, $CP[c]$, $TP$)
			\end{tabbing}
		\end{framed}
		\vspace*{-4mm}
		\caption{The \mindist{} algorithm.}
		\label{fig:algomin}
	\end{center}
	\vspace*{-6mm}
\end{figure}

Any object pair $(\oi, \oj)$ having a lesser distance than $(O_g, O_h)$ must
have a term pair $(\ti \in \oi, \tj \in \oj)$ which has a lesser distance than
any term pair $(t_g \in O_g, t_h \in O_h)$.  Hence, we only need to identify
such term pairs $(\ti, \tj)$ that are close and process their inverted lists,
i.e., the list of objects.

A pre-processing step is required to build the inverted lists of objects at each
node.  The inverted index is needed to be built for the minimum and average
pairwise distances but not the earth mover's distance.  For each object \oi,
when a term \tj{} appears in it, \oi is inserted into the inverted list of \tj.
The list is accessed using hashing, and the object is
inserted at the top of list.

Figure~\ref{fig:algomin} describes the entire algorithm.  For a node, the
\mindist{} algorithm computes the top-$k$ object pairs that are described by at
least one term pair in its subtree.  Any such object pair must either (i)~be in
the top-$k$ list of the children, or (ii)~contains terms from different subtrees
of the children.  The recursive definition of the first kind allows us to employ
a divide-and-conquer approach.  For the second kind, we need a list of objects
that are close to the subtree of the children nodes.  The lists can then be
joined to generate the necessary object pairs.  Thus, the \mindist{} algorithm
computed at the root of the ontology returns the top-$k$ pairs.

As shown in Figure~\ref{fig:algomin}, each node $t$ maintains two lists: (i)~a list
of pairs of objects $P$ ordered by their \dmin{} distances; and (ii)~a list of
objects $B$ ordered by their minimum distances $d_i$ to the node $t$.  The
length of $P$ is at most $k$.  The length of $B$ should be enough to ensure that
$k$ distinct pairs of objects can be generated from $B$.  The number of terms
required to do that is $k' = O(\lceil \sqrt{k} \rceil)$.\footnote{Since
$k'(k'-1)/2 \geq k$, the actual number of terms required is $k' = \lceil 1/2 +
\sqrt{1/4 + 2k} \rceil$.}

When \mindist{} is called on a node $t$, it selects $k'$ objects from
its list $L$ into $TB$.  $L$ is the list of objects associated with that term.
\mindist{} is then called on each of its $c$ children.  The cost of the edge
from $t$ to its child is added to the objects in the corresponding child's
object list (line 5).  This is done to ensure that the distances are maintained
correctly.  The $c$ sorted object lists and the list of objects in $t$ are then
merged to produce the sorted list $B$.

The merging (line 7) is done using a heap data structure~\cite{tr01}.  The heap
is initialized with $c + 1$ elements at position $1$ of each of the child lists
and the list $TB$.  The minimum element is then extracted into $B$.  Since all
the individual lists are sorted, the properties of heap guarantee that the
object extracted has the least \dmin{} distance from this node.  The object at
the next position of the list from where this minimum object came is then
inserted into the heap.  This is repeated $k'$ times.  

All the possible $k$ pairs are then generated from the $k'$ objects in $B$
(method GenPairs in line 8).  This list $TP$ computes the $k$ best distances of
the object pairs which are not in any of the subtrees.

$TP$ is finally merged with the pair lists $CP[i], i = 1 \dots c$ from the
children to produce the final pair list $P$ using a heap in the same manner as
above (line 9).

\subsection{Analysis of Time Complexity}
\label{sec:minanal}

In this section, we analyze the time and space complexities of the
\mindist{} algorithm.

We first analyze the time required to compute the inverted index.  The object
descriptions are read once, and for each term in an object, the corresponding
list is accessed in $O(1)$ time using hashing, and the object is inserted at the
top of list in another $O(1)$ time.  The total time required for this phase is,
therefore, $O(D)$.

We next analyze the running time for the main phase of the algorithm.  Selecting
$k'$ objects in $TB$ requires $O(k')$ time.  Adding the child edge costs to each
object in $CB$ lists takes $O(k' c)$ time.

At every step of the merging operation, the object with the minimum distance is
extracted from the heap and another object is inserted.  The size of the heap
is, therefore, never more than $O(c)$.  Extracting the minimum element and
inserting another object into the heap takes $O(\log c)$ time.  Since the
operation is repeated $k'$ times, the total running time of the merging
procedure is $O(k' \log c)$.

If, however, the objects in the child lists are not unique, $k'$ operations may
not be enough to select $k'$ different objects.  Thus, a hashtable is used to
ensure that an object is inserted into the heap only once.  First, all the lists
are scanned in $O(k' c)$ time.  If an object appears for the first time, it is
inserted into the hashtable with the object identifier as the key and the
distance as the value.  If an object appears twice, the one with the minimum
distance is maintained in the hashtable.  Before any object is inserted into the
heap, the hashtable is checked.  If this object is different from the one
maintained in the hashtable, then there exists another copy of this object with
a smaller distance.  Hence, this object does not need to be considered.  This
limits the number of heap operations to $O(k')$.  Assuming that the hashtable
operations take constant time, the running time then is $O(k' \log c)$.

Sorting $k$ local object pairs requires $O(k \log k)$ time.

Finally, the sorted pair lists at the node and the children are merged in $O(k
c + k \log c)$ time using a heap and a hashtable in a similar manner as before.

Thus, the total running time of the \mindist{} algorithm at a node with $c$
children is $O(k c + k \log k)$.

The algorithm is run once at each node of the ontology.  Assuming that there are
$T$ terms in the ontology, the total number of children for \emph{all} the nodes
is $O(T)$.  Hence, the amortized cost is $O(T k + T k \log k) = O(T k \log k)$.

The total running time of the \mindist{} algorithm is, therefore, $O(D + T k
\log k)$.
\\

\noindent
\textbf{Space Complexity:}
Each node in the ontology contains an object list of size $O(k')$ and a pair
list of size $O(k)$.  Once these lists are sent to the parent, they are no
longer required.  Thus, at any time, the space requirement at a node is $O(c (k'
+ k))$.  The total space complexity, therefore, is $O(c_{max} (k + k'))$ where
$c_{max}$ is the largest branching factor of a node in the tree.  The inverted
index requires $O(D)$ space for storage.

\section{The Algorithm for \avgdist{}}
\label{sec:avg}

Unlike the \mindist{} algorithm that needs to maintain only one term pair for
each object pair, the \davg{} distance needs to remember all the possible
term-pair distances.  Consequently, it runs in two phases: (i)~the \emph{Build}
phase, when pertinent information about objects are collected at the root in a
bottom-up manner, and (ii)~the \emph{Query} phase, when such information is used
to identify the top-$k$ pairs in a top-down order.

For any pair of objects, there are two types of costs that need to be
accumulated.  The first is the \emph{across-tree} costs, i.e., the distances
between the describing terms that occur in different subtrees of the root, and
the second is the \emph{within-tree} costs, i.e., the distances between the
describing terms that are within the same child of the root.  For example, in
Figure~\ref{fig:example}, the total pairwise term distances for ($O_1, O_4$) can be
broken into 2 parts: (i)~the across-tree distances between $t_1$ of $O_1$ and
$t_6$, $t_8$ of $O_4$ in the different subtrees under $t_1$ and $t_2$
respectively, and (ii)~the within-tree distances between $t_7$ of $O_1$ and
$t_6$, $t_8$ of $O_4$ in the same subtree under $t_2$.

To estimate the across-tree distances for object pairs at a node, the following
information need to be calculated for each object: (i)~the \emph{number of
terms} in the subtree that describe the object, and (ii)~the \emph{total
distance} of all such terms to this node.  This information is accumulated at
the root of the ontology by the build phase \avgdistbuild{}, which we describe
in Section~\ref{sec:avgbuild}.  After this phase, the root has collected the
following tuple for each object: ($\oi$, $n_i$, $w_i$).  

Before describing the two phases of the algorithm, we explain how lower bounds
for the across-tree costs of an object pair can be computed using the above
information and how such lower bounds can be generated in an \emph{ordered}
manner.

\subsection{Lower Bounds for Across-Tree Costs}
\label{sec:lb}

The estimates of the across-tree distances of a pair of objects \oi{} and \oj{}
at a node $t$ depend on the occurrences of their describing terms.  The
\emph{span} of an object is defined to be the number of subtrees of the root
where its constituent terms occur.  It can be either single, i.e., its terms
occur in only one subtree, or multiple, i.e., its terms occur in multiple
subtrees.  Based on these, 3 different cases need to be considered.  In each
case, we would like to write the bounds at a node in terms of the parameters
maintained for \oi{} and \oj{} at the node, i.e., in terms of ($\oi$, $n_i$,
$w_i$) and ($\oj$, $n_j$, $w_j$).\\

\noindent
\textbf{Case 1:} Both the objects have single spans.  Two sub-cases need to be
considered.\\

\noindent
\textbf{Sub-case 1(a):} The objects are in the same subtree.  The
across-tree cost is $0$ and nothing can be concluded about their distance in the
subtree without descending deeper into the subtree.  Hence, the lower bound is
\begin{equation}
	\label{eq:lb1a}
	\dlb = 0
\end{equation}

\noindent
\textbf{Sub-case 1(b):} The objects are in different subtrees.  In this case,
the across-tree distance can be estimated exactly.  The distance between a term
$t_i \in O_i$ and $t_j \in O_j$ is $d(\ti, t) + d(\tj, t)$ where $t$ is the node
at which this lower bound is being computed.  The \emph{total} across-tree
distance is obtained by adding all such combinations of terms:
\begin{align}
	\label{eq:sum1b}
	& \sum_{j=1}^{|\oj|} \sum_{i=1}^{|\oi|} d(\ti, t) + \sum_{i=1}^{|\oi|}
	\sum_{j=1}^{|\oj|} d(\tj, t) \nonumber \\
	= & n_j . \sum_{i=1}^{|\oi|} d(\ti, t) + n_i .
	\sum_{i=1}^{|\oi|} d(\tj, t) \nonumber \\
	= & n_j . w_i + n_i . w_j 
\end{align}
Thus, the average distance is
\begin{equation}
	\label{eq:lb1b}
	\dlb = d = \frac{w_i}{n_i} + \frac{w_j}{n_j}
\end{equation}
Since the within-tree distance for this pair is $0$, this is
the exact distance.\\

\noindent
\textbf{Case 2:} Both the objects have multiple spans.  The minimum across-tree
distance can be estimated in a manner similar to that in Case 1(b).  There are
at least two pairings of terms of $O_i$ and $O_j$ that are in different
subtrees.  Using Eq.~\eqref{eq:sum1b}, the \emph{total} across-tree costs for
these pairings are $w_{i_1} . n_{j_2} + w_{j_2} . n_{i_1}$ and $w_{i_2} .
n_{j_1} + w_{j_1} . n_{i_2}$, where $n_{i_1}$, $w_{i_1}$ etc. are the number of
terms of \oi{} in one subtree and its total distance to $t$ from that subtree.
The values of $n_{i_1}$, $n_{i_2}$, $n_{j_1}$, and $n_{j_2}$ are at least $1$.
Thus, the total across-tree distance is at least $w_{i_1} + w_{i_2} + w_{j_1} +
w_{j_2} = w_i + w_j$.  The lower bound for the average pairwise distance, then,
is 
\begin{equation}
	\label{eq:lb2}
	\dlb = \frac{w_i + w_j}{n_i . n_j}
\end{equation}

\noindent
\textbf{Case 3:} One object \oi{} has a single span, and the other object \oj{}
has a multiple span.  Similar to Case 2, there is at least one subtree
containing terms of $O_j$ but not containing terms of $O_i$.  The \emph{total}
across-tree cost is then the minimum of $w_i . n_{j_1} + w_{j_1} .  n_i$ and
$w_i . n_{j_2} + w_{j_2} . n_i$.  Similar to Case 2, there is at least one term
of \oj{} that is not in the same subtree of \oi{}.  Thus, $n_{j_1}$ and
$n_{j_2}$ are at least $1$.  However, without knowing where the terms of \oj{}
occur, nothing can be concluded about $w_{j_1}$ and $w_{j_2}$.  Since the terms
may occur at the node itself, the estimates for $w_{j_1}$ and $w_{j_2}$ are $0$.
Hence, the total distance is at least $w_i$ producing a lower bound of
\begin{equation}
	\label{eq:lb3}
	\dlb = \frac{w_i}{n_i . n_j}
\end{equation}

\subsection{Generating Ordered Pairs}
\label{sec:avgorder}

Though the above mentioned lower bounds can be computed for a given pair, the
cost of computing them for every pair is $O(N^2)$.  We would like to avoid such
costly online operations.  The trick is to separate the parameters of \oi{} and
\oj{} in each lower bound such that they can be systematically generated in an
\emph{ordered} manner whenever needed.  The order of generation will guarantee
that at any point of time, the lower bounds of the pairs not examined will be
greater than or equal to the lower bounds of the pairs already generated.  In
this section, we will discuss ways to achieve this for each of the cases
mentioned above.

To identify pairs of objects in the same subtree (Case 1(a)), $c+1$ different
lists are maintained at the root corresponding to itself and its $c$ children.

To handle Case 1(b), each of these $c + 1$ lists of objects are sorted by the
average distance $w_i / n_i$.  Given two such sorted child lists, it is
guaranteed that the lower bound (which is the sum of the distances) for an
object pair at positions $p_i$ in the first list and $p_j$ in the second list is
lower than the estimate of every pair whose positions are $> p_i$ and $> p_j$.
Thus, every time a pair at positions $(p_i, p_j)$ is inspected, only its
immediate successors $(p_i+1, p_j)$ and $(p_i, p_j+1)$ need to be considered.
Since there are $c + 1$ child lists, the number of possible ways of pairing is
$c(c + 1)/2$.

The lower bound for Case 2 is not easily separable in terms of parameters of
\oi{} and \oj{}.  It is, however, separable if for an object pair, the number of
terms for the objects (i.e., $n_i$, $n_j$) are known a priori.  To do that, the
list of objects with multiple spans is partitioned such that each partition
contains objects with a particular $n_i$.  Pairing \oi{} and \oj{} and knowing
which partitions they come from immediately defines the denominator of the lower
bound.  Thus, if there are $r$ partitions, sorting each partition by $w_i$ and
performing $r^2$ pairings in the same way as done for Case 1(b) orders the pairs
according to their lower bounds.

Case 3 is handled similarly.  The single-span list is broken into $c + 1$ lists
and the multiple-span list into $r$ partitions.  Generating all $r(c + 1)$
pairings gives the lower bounds in an ordered manner.

We next describe how the Query phase of the \avgdist{} algorithm uses these
lower bounds.

\subsection{Query Phase}
\label{sec:avgquery}

\begin{figure}[t]
	\begin{center}
		\begin{framed}
			\begin{tabbing}
				Algorithm \textbf{AvgDist-Query} \\
				\textbf{Input:} Node $root$ \\
				\textbf{Output:} Pair list $P$ of size $k$ \\
				1. $L$ := list of object mappings in $root$ \\
				2. $P$ := $\Phi$ (therefore, $P.dist$ := $\infty$) \\
				3. $c$ := number of children of $root$ \\
				4. $r$ := number of partitions of objects \\
				5. Divide $L$ into $c + 1 + r$ lists $B$ \\
				6. $A$ := GenInitialPairs($B$) \\
				7. Insert each $a \in A$ into heap $H$ \\
				8. $p$ := Pop($H$) \\
				9. \textbf{while} \= $p.dist < P.dist$ \\
				10. \> \textbf{if} \= Done($p$) = false \\
				11. \> \> $p$ := UpdateEstimate($p$) \\
				12. \> \textbf{end if} \\
				13. \> \textbf{if} \= Done($p$) = true \\
				14. \> \> \textbf{if} \= $p.dist < P.dist$ \\
				15. \> \> \> Insert $p$ into $P$ \\
				16. \> \> \> $P.dist$ := \kth{} distance in $P$ \\
				17. \> \> \textbf{end if} \\
				18. \> \textbf{else} \\
				19. \> \> Insert $p$ into $H$ \\
				20. \> \textbf{end if} \\
				21. \> $A$ := GenNextPairs($p$, $L$) \\
				22. \> Insert each $a \in A$ into $H$ \\
				23. \> $p$ := Pop($H$) \\
				24. \textbf{end while}
			\end{tabbing}
		\end{framed}
		\vspace*{-4mm}
		\caption{The \emph{Query} phase of the \avgdist{} algorithm.}
		\label{fig:algoavgquery}
	\end{center}
	\vspace*{-6mm}
\end{figure}

The \avgdistquery{} procedure (Figure~\ref{fig:algoavgquery}) is run at the root of
the ontology.  It outputs a list $P$ of top-$k$ object pairs.  When the size of
$P$ is less than $k$, $P.dist$ is $\infty$; otherwise, it is maintained as the
\kth{} largest distance in $P$.

The list $L$ of objects is broken into $c + 1 + r$ lists corresponding to single
and multiple spans as explained in the earlier section.  From these lists, the
initial pairs with the lower bounds are generated (method GenInitialPairs in
line 6) and put into a heap $H$.  See Section~\ref{sec:avgorder} for details on
how to generate these pairs.

The top-down searching for object pairs proceeds in a manner where at every
stage, only the current ``best'' pair is examined~\cite{re38}.  Thus, this
search strategy is called the \emph{best-first search}.

The algorithm progresses by extracting the current \emph{best pair} from the
heap, i.e., the pair $p$ with the current best lower bound (line 8).  If the
lower bound is an estimate for $p$ and not an exact distance as in Case 1(b),
the bound can be improved in two ways (line 11).  First, the within-tree costs
at the subtrees in the next level can be estimated again using
Eqs.~(\ref{eq:lb1a}-\ref{eq:lb3}) by descending into the subtree (denoted as
\avgdist-\nextestimate).  The descent is made in a breadth-first order on the
tree.\footnote{Any order, e.g., depth-first order, will also work.  However, if
the edge distances decrease exponentially, breadth-first ordering produces
better bounds.}  The second way is to compute the term-wise distances fully
without resorting to recursion (denoted as \avgdist-\complete).  This, however,
disregards the structure of the ontology.

If the exact distance of $p$ is computed, the list $P$ is examined.  If the
\kth{} distance in $P$ is more than that of $p$ (line 14), $p$ is inserted into
$P$ and $P.dist$ is modified.  The size of $P$ is maintained to be at most $k$
by removing the pair with the largest distance.  

If, however, the lower bound of $p$ is still an estimate, $p$ is re-inserted
back into the heap $H$ (line 19).  The next pairs are generated from the $c + 1
+ r$ lists (method GenNextPairs in line 21 as described in
Section~\ref{sec:avgorder}) and inserted into the heap (line 22).

In the next iteration, the pair which is now the \emph{best} is examined (line
23).  If this pair has a distance more than the \kth{} distance in $P$ (i.e.,
$P.dist$), it is guaranteed that all the pairs currently in the heap and all the
pairs that are not generated will have a greater distance.  This is due to the
properties of the heap and the ordered nature of generating the pairs from the
$c + 1 + r$ lists.  Thus, the algorithm is then terminated correctly.

\subsection{Build Phase}
\label{sec:avgbuild}

In this section, we describe how \avgdistbuild{} computes the information
($\oi$, $n_i$, $w_i$) for an object.\footnote{Figure~\ref{fig:algoavgbuild} in
Appendix outlines the algorithm.}  Each node $t$ maintains an inverted list $L$
of objects \oi{} described using $t$.  First, it converts $L$ into $B$
by making $n_i = 1$ and $w_i = 0$ for each $\oi \in L$.  Then, it calls
\avgdistbuild{} for each of its children.  For each list $CB$ that it receives
from a child, and for each object $\oj \in CB$, it modifies $w_j$ by adding to
it the distance to the child node multiplied by the number of times \oj{} occurs
in the child subtree, i.e., $w_j = w_j + dist \times n_j$, where $dist$ is the
edge distance from $t$ to its child.  This ensures that the total distance from
$t$ is maintained correctly, since each of the $n_j$ objects have to traverse
the distance $dist$.
\\

\noindent
\textbf{Analysis of Space and Time Complexities:}
Assume the total size of the object description to be $D$ which is at most $N
\times T$ where $N$ is the total number of objects, and $T$ the total number of
terms.  The inverted index requires $O(D)$ time and space to construct.  We next
analyze the space and time complexity of \avgdistbuild{} in terms of these
parameters.

Each object's information is stored at the terms describing it.  The information
stored in a term is repeated along all its ancestors.  Since the size of the
description is $D$, and there are $O(\log T)$ ancestors (assuming the ontology
to be balanced), the storage cost is $O(D \log T)$.

The running time can be analyzed similarly.  At the leaf level of the tree,
there are $D$ describing terms.  When this $O(D)$ information is sent up to the
next level, the time required to combine the information is still $O(D)$ since
each object description is read only once and is matched using a hashtable to
the information already computed.  Assuming the height of the tree to be $O(\log
T)$, the total running time is $O(D \log T)$.

\section{Experiments}
\label{sec:expts}

\subsection{Datasets}
\label{sec:dataset}

\begin{table}
	\begin{center}
		\begin{tabular}{|c||c|c|}
			\hline
			Name & Number of & Number of \\
			& GO Terms ($T$) & Genes ($N$) \\
			\hline
			Process & 13762 & 3437 \\
			Function & 7803 & 1958 \\
			Localization & 1990 & 645 \\
			\hline
		\end{tabular}
		\caption{The Gene Ontology (GO) datasets.}
		\label{tab:go}
	\end{center}
	\vspace*{-7mm}
\end{table}

We have experimented with real as well as synthetic datasets.  The real dataset
is that of Gene Ontology (GO, \url{http://www.geneontology.org/}).  There are
three ontologies in GO, corresponding to biological process, molecular function
and cellular component (localization) of terms.  The details of the three
ontologies are given in Table~\ref{tab:go}.  The datasets were curated by
hashing gene descriptions using their bit-vector representations of the terms
and removing the identical genes.

The synthetic datasets were generated by controlling the number of objects, the
number of terms, the average branching factor of the ontology tree and the
average number of terms per object.  The ontologies and the object datasets are
created separately.  Ontologies have a fixed size and an average branching
factor.  Starting from the root, we generate a random number of children by
perturbing the average branching factor within some limits.  We continue with
this at all successive nodes.  The object dataset is generated with a fixed
number of objects and an average number of terms per object.  Again, a random
number is generated from the average by perturbing it.  Then, terms are picked
from the ontology randomly without replacement for the required number of terms.
This process is repeated for all objects.

\subsection{Experimental Setup}
\label{sec:setup}

When the distance function between the objects is defined as the \emph{earth
mover's distance} the following schemes were evaluated:
\begin{itemize}
	\item \lone{}-reduced: In this scheme (Section~\ref{sec:emd}), the \lone{}
		on reduced number of terms is used.
	\item \lone{}-full: In this scheme, the \lone{} on all terms is used.  The
		tree is not pruned at a height $1$.
	\item EMD-reduced: All the $O(N^2)$ EMDs on reduced number of terms are
		computed.  These are then used to prune those object pairs for which the
		reduced EMD is greater than the $k^\text{th}$ best EMD already found.
	\item Brute-force: In this scheme, all the $O(N^2)$ pairs are computed and
		then the top-$k$ pairs are returned.
\end{itemize}
The performance of the brute-force scheme (267\,s for $N=100$ objects) is too
impractical to be of any use and are, therefore, not reported.  Also, the times
of \lone{}-full are not reported since, in the best case, it can only save
\lone{}-reduced computations, which are very fast anyway.  In all the
experiments, it was actually worse than \lone{}-reduced.

When the distance function between the objects is defined as the \emph{minimum
pairwise distance} between the terms, the following schemes were considered:
\begin{itemize}
	\item \mindist: This is the scheme described in Section~\ref{sec:min} that
		has a running time of $O(T k \log k)$.
	\item Brute-force: In this scheme, all the $O(N^2)$ pairs are computed and
		then the top-$k$ pairs are returned.  Maintaining a heap of size at most
		$k$ gives the running time of this scheme to be $O(N^2 \log k)$.  Due to
		the exorbitant online costs of it, this scheme is not practically
		useful.
\end{itemize}
For $N=10^4$, the top-$k$ computation using the brute-force algorithm finishes
in $\sim$300\,s.  Since the \mindist{} has a better running time, we report the
experiments for \mindist{} only.

When the distance function between the objects is defined as the \emph{average
pairwise distance} between the terms, the following schemes were evaluated:
\begin{itemize}
	\item \avgdist-\nextestimate: In this variant of \avgdist{}, the estimate
		for the best-pair is improved by progressively descending into the
		subtrees and estimating the across-tree costs at the roots of those
		subtrees.
	\item \avgdist-\complete: This is the other variant of \avgdist{} where the
		exact distance is computed at one go by computing all the pairwise term
		distances.
	\item Brute-force: In this scheme, all the $O(N^2)$ pairs are computed and
		then the top-$k$ pairs are returned.
\end{itemize}
The performance of the brute-force scheme (300\,s for $10^4$ dataset) is much
higher than that for \avgdist{} schemes.  Consequently, it is not discussed any
further.

Sections~\ref{sec:k_emd} to~\ref{sec:b_emd} report experiments on \emddist{}
while Sections~\ref{sec:k_min} to~\ref{sec:n_min} and~\ref{sec:k_avg}
to~\ref{sec:t_avg} report on \mindist{} and \avgdist{} respectively.

\subsection{Effect of k on \emddist{}}
\label{sec:k_emd}

\begin{figure}[t]
	\begin{center}
		\includegraphics[width=\figwidth]{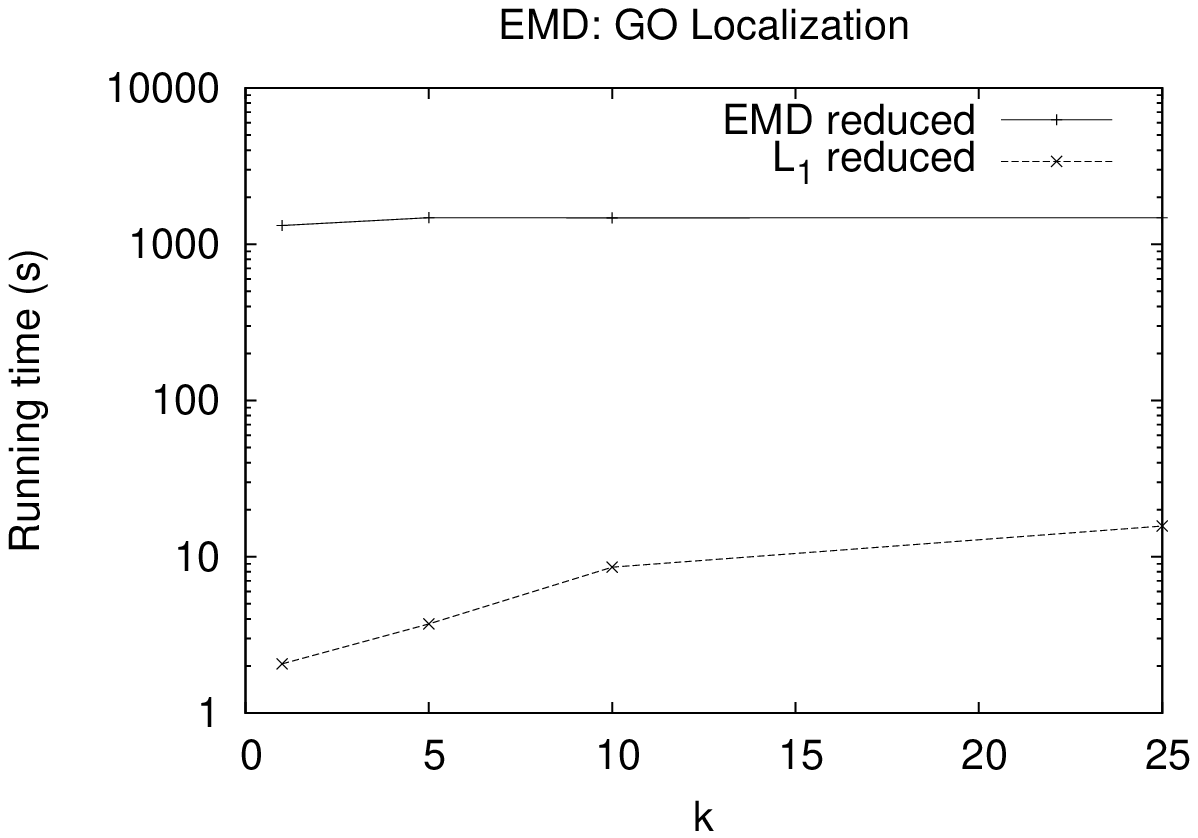}
		\vspace*{-2mm}
		\caption{Effect of $k$ on \emddist{}.}
		\label{fig:emd_go_k}
	\end{center}
	\vspace*{-6mm}
\end{figure}

Figure~\ref{fig:emd_go_k} shows the effect of $k$ on the running time of GO
localization dataset.  When $k$ is increased, more number of \lone{}
computations are needed before the TA can halt.  Consequently, more number of
\emddist{} calculations are also required.

\subsection{Effect of N on \emddist{}}
\label{sec:n_emd}

\begin{figure}[t]
	\begin{center}
		\includegraphics[width=\figwidth]{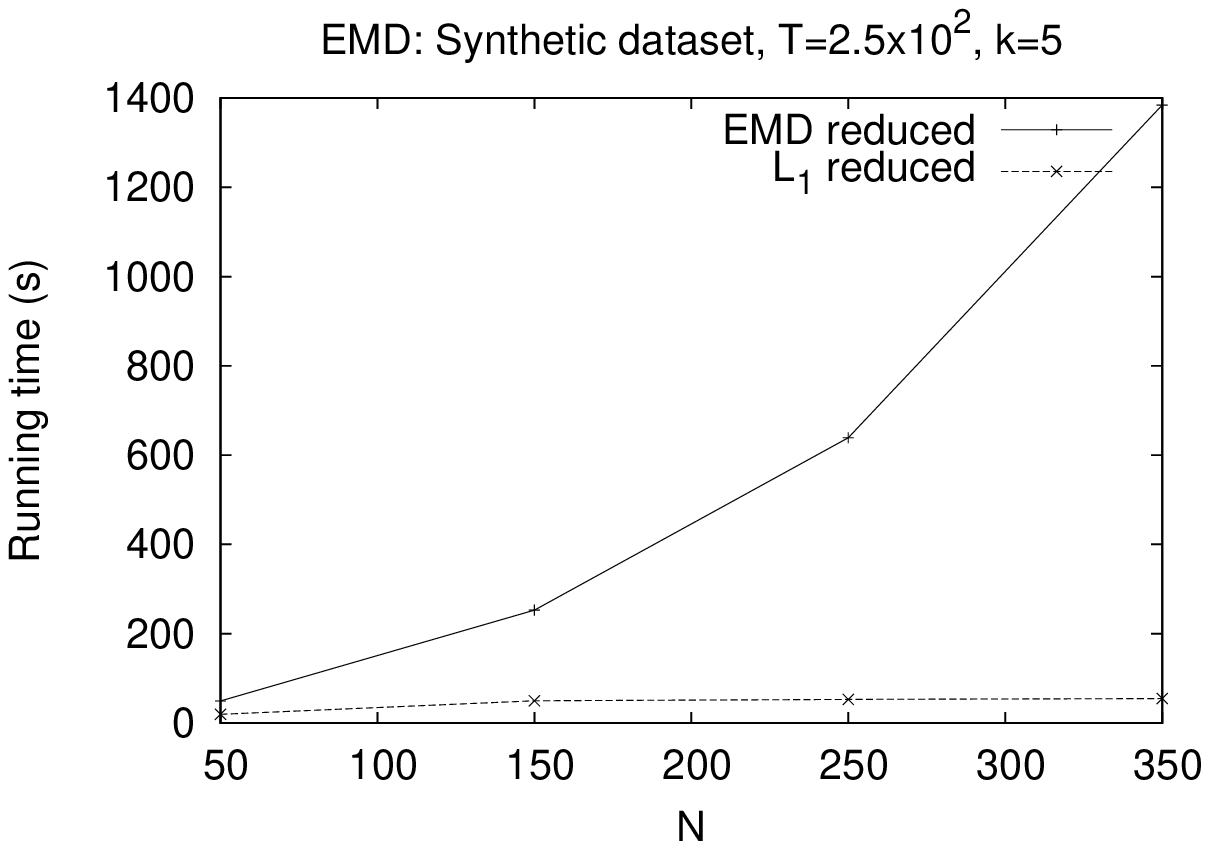}
		\vspace*{-2mm}
		\caption{Effect of $N$ on \emddist{}.}
		\label{fig:emd_n}
	\end{center}
	\vspace*{-6mm}
\end{figure}

Figure~\ref{fig:emd_n} shows that the scalability of our algorithm with $N$ is
better than quadratic.  Even though the number of objects increases
quadratically, due to \lone{} lower bounding, many of the object pairs are
pruned.  Consequently, the number of full \emddist{} computations increases by a
lower factor.  Also, even for $N=350$ which translates to $6 \times 10^4$ object
pairs, our algorithm finishes in only 55\,s.

\subsection{Number of Object Pairs for \emddist{}}
\label{sec:eta_emd}

\begin{figure}[t]
	\begin{center}
		\includegraphics[width=\figwidth]{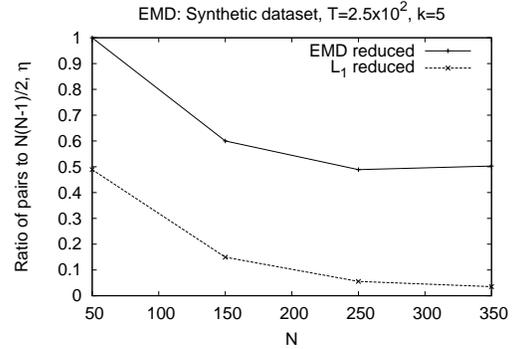}
		\vspace*{-2mm}
		\caption{Effect of $N$ on number of pairs examined for
		\emddist{}.}
		\label{fig:emd_eta}
	\end{center}
	\vspace*{-6mm}
\end{figure}

To check the effect of increasing $N$, we measured the ratio of object pairs for
which full \emddist{} computation was done.  The ratio was measured as number of
pairs investigated to the total number of possible pairs ($N(N-1)/2$) and is
denoted by $\eta$.  As Figure~\ref{fig:emd_eta} shows, $\eta$ decreases when $N$ is
increased.  For $N=250$, the number of \emddist{} computations becomes lower
than 10\,\%.

\subsection{Effect of T on \emddist{}}
\label{sec:t_emd}

The next experiment measures the effect of the total number of terms on the \emddist{}
computations.  Since both the \lone{} and EMD-reduced depends only on the
reduced number of terms, the effect of $T$ is minimal (graph not shown).

\subsection{Effect of t on \emddist{}}
\label{sec:b_emd}

\begin{figure}[t]
	\begin{center}
		\includegraphics[width=\figwidth]{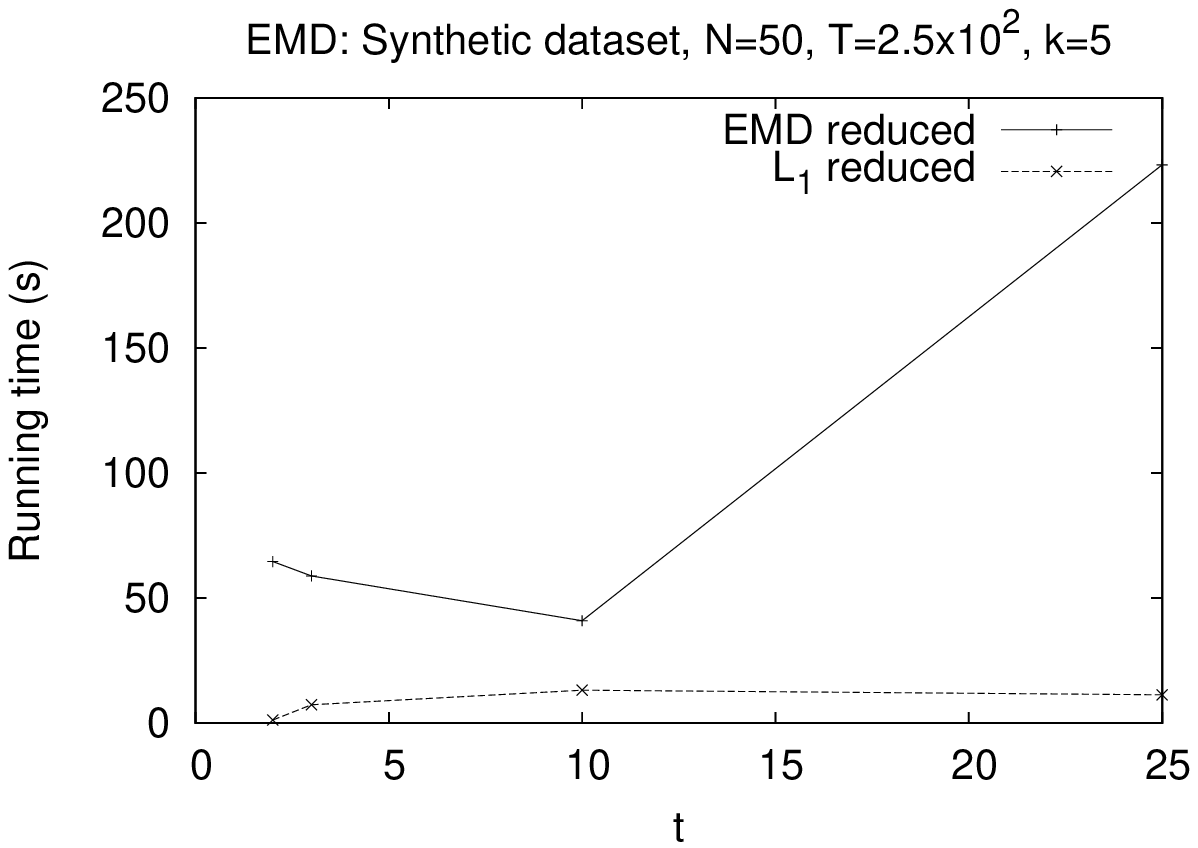}
		\vspace*{-2mm}
		\caption{Effect of $t$ on \emddist{}.}
		\label{fig:emd_b}
	\end{center}
	\vspace*{-6mm}
\end{figure}

As the number of children of root, i.e., $t$ increases, the complexity of the TA
increases linearly.  Figure~\ref{fig:emd_b} shows the running times for varying
$t$.  The size of each object description is limited to $10$.  When $t \leq 10$,
the time increases.  The EMD-reduced behaves in the opposite manner.  This is
due to the interaction of two opposing effects: as $t$ increases, each
computation takes more time, but the lower bound gets tighter as more number of
terms are taken into account resulting in less number of full \emddist{}
computations.  However, when $t > 10$, since there are at most $10$ terms in
each object, the object description size do not get reduced and each
\emddist{}-reduced computation takes as much time as the full \emddist{}
computation.  Since $O(N^2)$ of these computations are performed, the running
time shoots up.  The \lone{}-reduced, on the other hand, shows only a little
increase.

The next set of experiments measure the effect of different parameters on
the \mindist{} algorithm.

\subsection{Effect of k on \mindist}
\label{sec:k_min}

The first set of experiments measure the effect of the number of top pairs
queried ($k$), on the running time of the \mindist{} algorithm.  As shown in
Figure~\ref{fig:min_go_k}, the scalability of \mindist{} with $k$ is linear.  The
analysis done in Section~\ref{sec:minanal} shows that for small values of $k$,
this is the expected behavior.  The largest real dataset---GO process---finishes
in less than 1\,s for $k \leq 50$, demonstrating the effectiveness of the
algorithm.

\begin{figure}[t]
	\begin{center}
		\includegraphics[width=\figwidth]{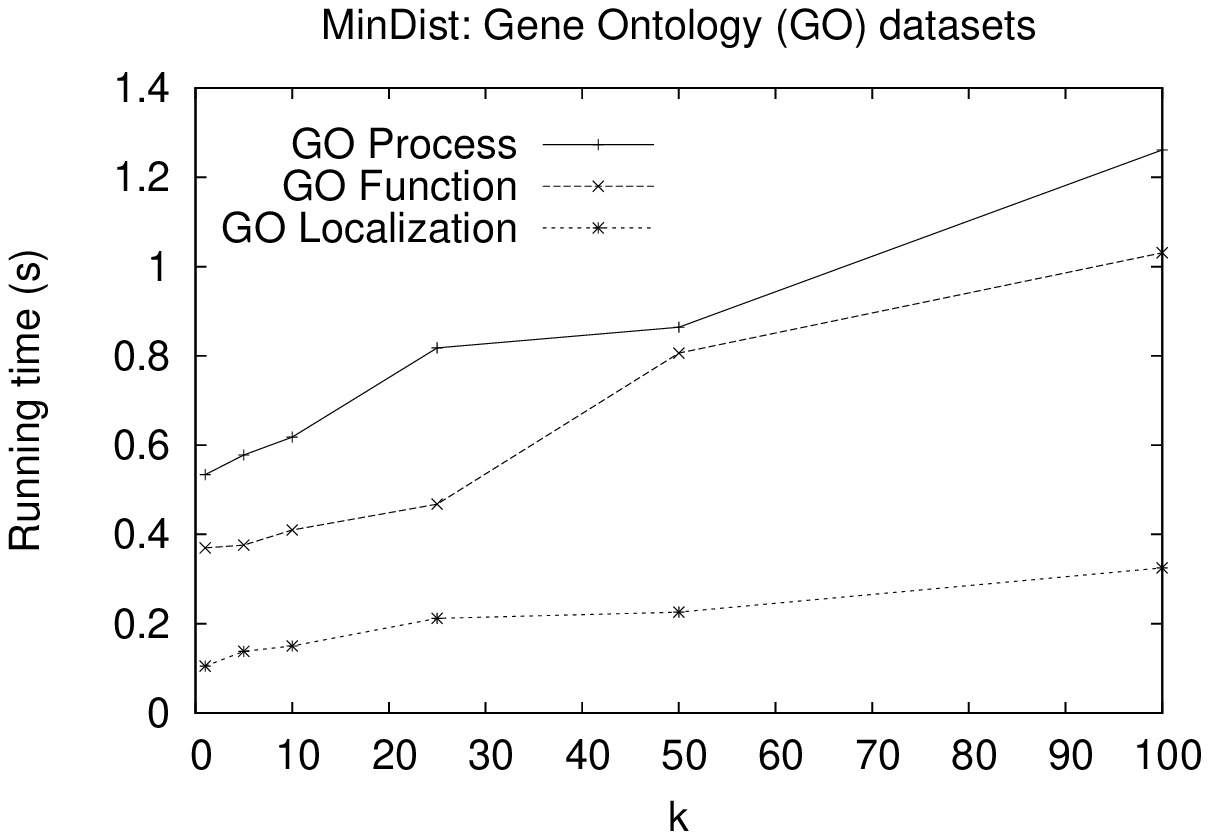}
		\vspace*{-2mm}
		\caption{Effect of $k$ on \mindist{}.}
		\label{fig:min_go_k}
	\end{center}
	\vspace*{-6mm}
\end{figure}

\subsection{Effect of T on \mindist}
\label{sec:t_min}

\begin{figure}[t]
	\begin{center}
		\includegraphics[width=\figwidth]{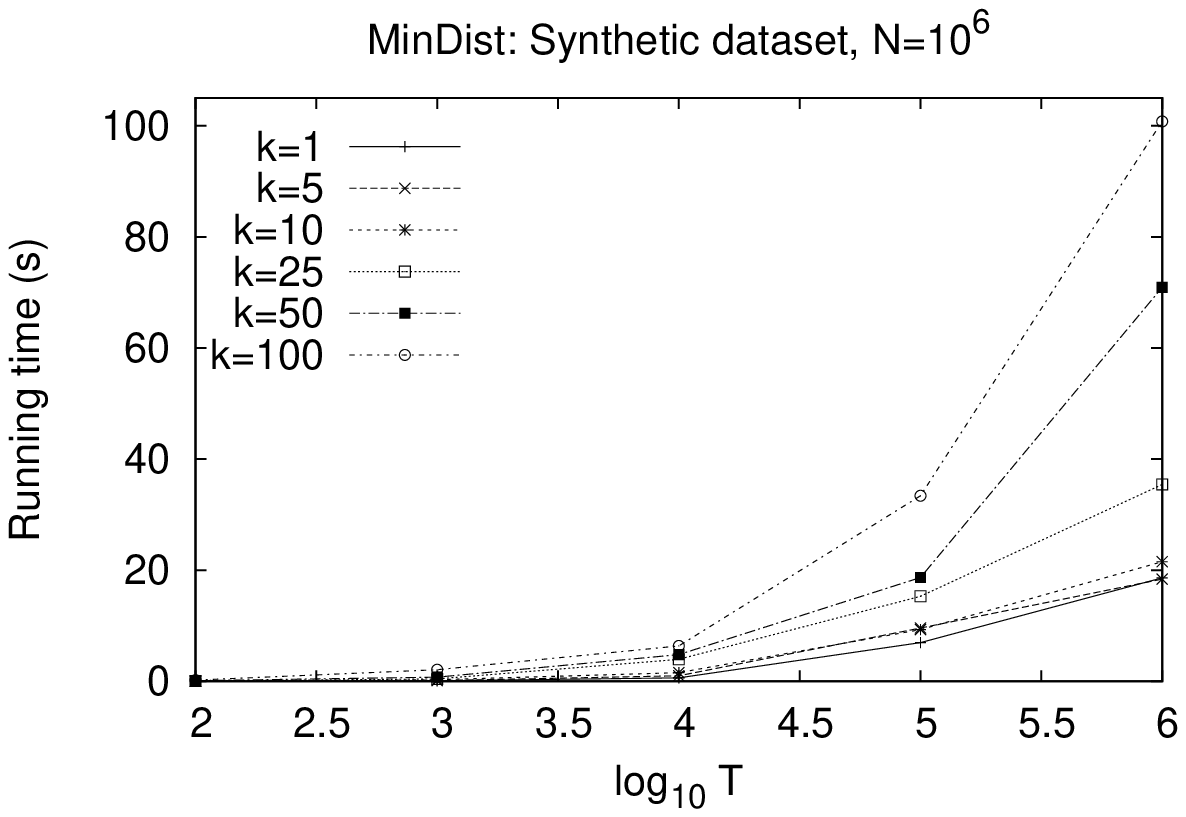}
		\vspace*{-2mm}
		\caption{Effect of $T$ on \mindist{}.}
		\label{fig:min_t}
	\end{center}
	\vspace*{-6mm}
\end{figure}

We next report the effect of the number of terms $T$ on the running time.
Figure~\ref{fig:min_t} shows that increasing $T$ increments the running time of
\mindist{} linearly, independent of the value of $k$.  We also note the
practicality of the \mindist{} algorithm.  For a very large dataset of size $N =
10^6$ and a very large tree of size $T = 10^6$, a top-$100$ query finishes in
about 100\,s.  For smaller $k$'s and for smaller $T$'s, the running time is in
seconds.

\subsection{Effect of N on \mindist}
\label{sec:n_min}

The running time analysis of the \mindist{} algorithm shows that it is
independent of the number of objects $N$.  When the number of terms $T$ is kept
constant, the experiments confirm that the running time is practically constant
even when $N$ is increased from $10^3$ to $10^6$ (graph not shown).

The next set of experiments evaluate the performance of the two variants of
\avgdist{}.

\subsection{Effect of k on \avgdist}
\label{sec:k_avg}

\begin{figure}[t]
	\begin{center}
		\includegraphics[width=\figwidth]{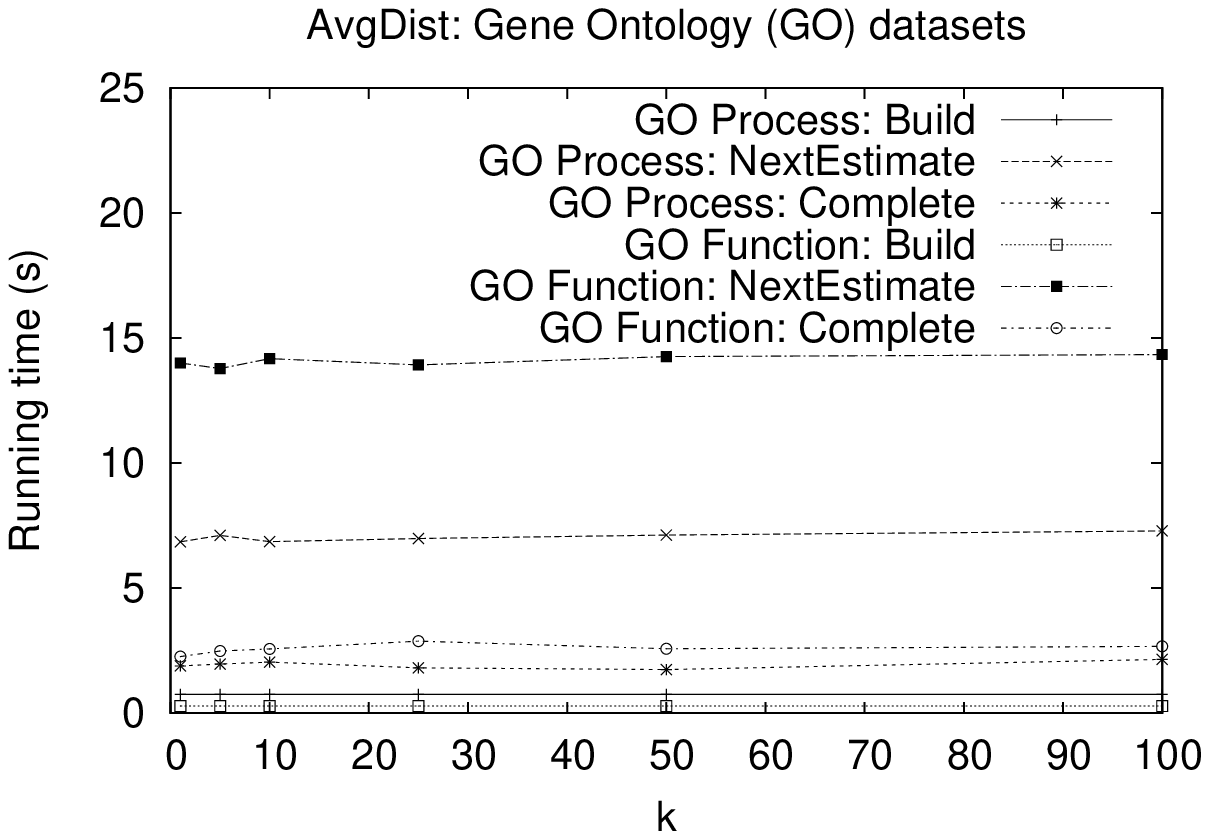}
		\vspace*{-2mm}
		\caption{Effect of $k$ on \avgdist{}.}
		\label{fig:avg_go_k}
	\end{center}
	\vspace*{-6mm}
\end{figure}

The first experiment on \avgdist{} illustrates the effect of $k$ on the running
time of the Build phase and the two different variants---\nextestimate{} and
\complete{}---for the two larger GO datasets.  All the six curves in
Figure~\ref{fig:avg_go_k} are relatively flat, showing that the effect of $k$ is
minimal.  Intuitively, the running time of \avgdist{} depends on the actual
number of object pairs investigated.  For the GO datasets, even for even large
$k$'s up to $100$, this remains almost constant.  Moreover, the Build phase
takes negligible time in comparison to the Query phase.

\subsection{Number of Object Pairs for \avgdist{}}
\label{sec:objpairs}

\begin{figure}[t]
	\begin{center}
		\includegraphics[width=\figwidth]{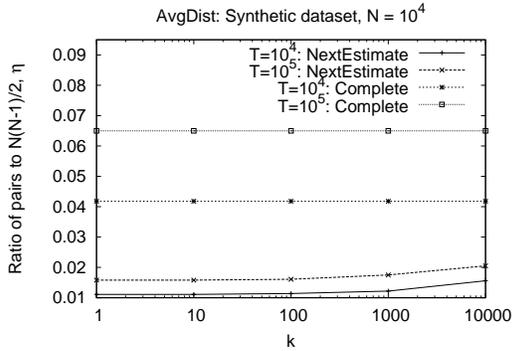}
		\vspace*{-2mm}
		\caption{Effect of $k$ on number of pairs examined for
		\avgdist{}.}
		\label{fig:avg_eta_k}
	\end{center}
	\vspace*{-6mm}
\end{figure}

We further investigated the effect of $k$ by measuring the number of object
pairs that are examined in the Query phase of the \avgdist{} algorithm.  For
this, we increased $k$ up to $10000$.  Figure~\ref{fig:avg_eta_k} shows that $\eta$
(i.e., the ratio to the total number of possible pairs) increases very slowly
with $k$.  The results are robust across different values of $T$ (as shown in
the figure) and $N$ (not shown).  This is the reason why the running time is
also constant across $k$.  

The \nextestimate{} method examines less than 2\% of the total number of pairs.
The \complete{} method investigates more object pairs (about 7\%) than the
\nextestimate{} method.  Computing a distance for the current best-pair
guarantees that only those pairs which have a bound lower than this distance
will be analyzed.  For the \nextestimate{} method, the distance of the best-pair
is computed progressively, thereby saving on full \avgdist{} computations as
compared to the \complete{} method, which finds the actual distance of the
best-pair.

\subsection{Effect of N on \avgdist}
\label{sec:n_avg}

We next discuss the experimental results when the number of objects
is varied.  We first measure the effect of number of objects on the Build phase.
From the analysis done in Section~\ref{sec:avgbuild}, we expect the running time
to grow linearly with the size of the input information.  Assuming that the
number of describing terms for an object is constant, the size of the
information is directly proportional to the number of objects.  The experiment
shows that the scalability is indeed linear (graph not shown).

\begin{figure}[t]
	\begin{center}
		\includegraphics[width=\figwidth]{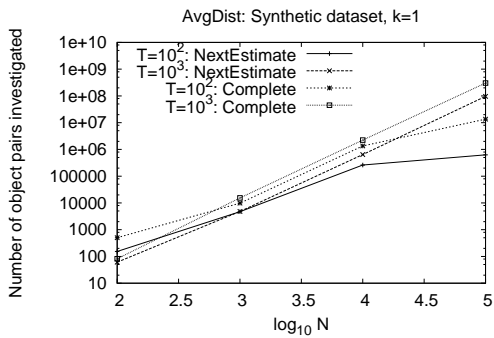}
		\vspace*{-2mm}
		\caption{Effect of $N$ on number of pairs examined for \avgdist{}.}
		\label{fig:avg_n_eta}
	\end{center}
	\vspace*{-6mm}
\end{figure}

The next experiment (Figure~\ref{fig:avg_n_eta}) shows that the number of pairs
investigated grows at most quadratically with $N$.  Since the objects are
generated using the same random process, this is expected.

\subsection{Effect of T on \avgdist}
\label{sec:t_avg}

The next set of experiments measure the effect of the number of terms $T$ on the
different components of the \avgdist{} algorithm.  Figure~\ref{fig:avg_t_build}
shows the time taken to complete the Build phase.  Note that this phase takes
the same amount of time regardless of the choice of the method for estimating
the distance of a pair.  Since the build procedure is run at each node, the
effect of $T$ is linear.  Further, as can be seen from the plot, when the number
of objects increase, more information needs to be processed at each node and the
running time increases linearly.

The next experiment 
measures the number of pairs investigated against different values of $T$.  As
shown in Section~\ref{sec:n_avg}, the number of pairs depends primarily on the
distribution of the objects on the tree---mainly the number of objects falling
in the single span lists---and not on the size of the tree.  Consequently, the
size of the tree $T$ has no appreciable effect.  Similar to the previous set of
experiments, this effect of $T$ (or rather the lack of it) is directly reflected
in the running time as well.  The running time is essentially independent of $T$
(graph not shown).

\section{Conclusions}
\label{sec:concl}

In this paper, we proposed the problem of finding top-$k$ most similar object
pairs annotated with terms from an \emph{ontology}.  The terms represent
concepts and the objects are described using these concepts.  The join problem
exposed the computational aspects of the domain well.

We then defined and motivated three object distances that can be used to define
the dissimilarity (or, equivalently similarity) between a pair of objects.  The
\emph{minimum pairwise distance} is useful in order to search objects that share
a similar term.  The \emph{average pairwise distance} captures the notion of
similarity when the object definitions are imprecise or when objects need to be
compared on multiple attributes.  The third one, \emph{earth mover's distance},
is particularly useful as it finds the best way of matching terms in one object
with those in the other by capturing the term-to-term relationships, and
measures the distance corresponding to this best matching.

Finally, we designed algorithms to efficiently solve the problem using all the
above distance measures.  The algorithm for EMD uses \lone{} distance as a lower
bound and even avoids all \lone{} computations by modifying the threshold
algorithm.  The algorithm that solves the problem for the minimum pairwise
distance runs in $O(D + T k \log k)$ time.  For the average pairwise distance,
we devised a best-first search strategy that avoids all pairs investigation by
generating lower bounds in an ordered manner.  Experimental evaluations
demonstrated the practicality and scalability of our algorithms.

In future, we would like to design algorithms for other distance measures and
lower bounds.  We would also like to develop methods that use term statistics to
improve the expected running time and further explore the optimal height of
pruning the ontology tree for EMD.  Lastly, algorithms for $k$-NN and range
queries should be simple extensions of the proposed algorithms.

\begin{figure}[t]
	\begin{center}
		\includegraphics[width=\figwidth]{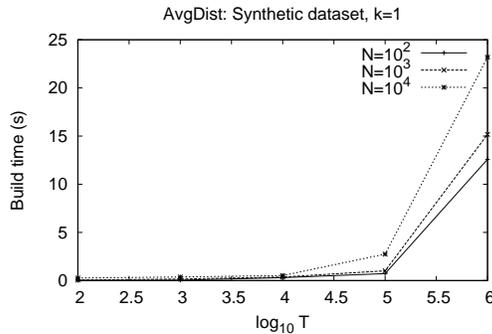}
		\vspace*{-2mm}
		\caption{Effect of $T$ on Build phase of \avgdist{}.}
		\label{fig:avg_t_build}
	\end{center}
	\vspace*{-6mm}
\end{figure}

{
\bibliographystyle{abbrv}
\bibliography{onto}
}

\pagebreak

\appendix

\section*{Appendix}

\section{Average pairwise distance follows triangular inequality}

\begin{lemma}
	\label{lemma:avg}
	The average pairwise distance \davg{} as defined in Eq.~\eqref{eq:avg}
	follows the triangular inequality property.
\end{lemma}
\begin{proof}
	Assume any three objects $A$, $B$ and $C$.  We need to prove that $\davg(A,
	B) + \davg(B, C) \geq \davg(C, A)$.
	
	Consider any term $a_i \in A$, $b_j \in B$, and $c_k \in C$.  Since the term
	distance function is a metric, we can write $d(a_i, b_j) + d(b_j, c_k) \geq
	d(c_k, a_i)$.
	Adding the $|A|.|B|.|C|$ equations together yields
	\begin{align}
		%
		%
		%
		\sum_{i,j,k=1}^{|A|,|B|,|C|} d(a_i, b_j) + \sum_{i,j,k=1}^{|A|,|B|,|C|}
		& d(b_j, c_k) \nonumber \\
		& \geq \sum_{i,j,k=1}^{|A|,|B|,|C|} d(c_k, a_i) \nonumber \\
		\textrm{or, } |C|.\sum_{i,j=1}^{|A|,|B|} d(a_i, b_j) + 
		|A|.\sum_{j,k=1}^{|B|,|C|} & d(b_j, c_k) \nonumber \\ 
		& \geq |B|.\sum_{k,i=1}^{|C|,|A|} d(c_k, a_i) \nonumber 
	\end{align}
	Dividing by $|A|.|B|.|C|$, we get
	\begin{equation}
		\davg(A, B) + \davg(B, C) \geq \davg(C, A) \nonumber
	\end{equation}
	\hfill{}
\end{proof}

\section{Hashing}

If \lone{} is computed on all the terms in the TA phase of the \emddist{}
algorithm, then the time required for sorting of $N$ objects in the initial
phase can be saved.  The key is to observe that all values for an object will be
of the form $1/c$ where $c$ is the count of the number of terms in the object.
Since $c$ is at most $T$, a hashtable of size $T$ with keys $\frac{1}{1}, \dots,
\frac{1}{T}$ can be maintained.  The $N$ object values will be hashed into it.
The heap $H$ will be filled up with values of the form $\frac{1}{i} -
\frac{1}{i+1}$ only.  This requires a running time of $O(N+T)$ instead of $O(N
\log N)$.

When reduced number of terms are used, the values will be of the form
$\frac{t_1}{t_2}$, where $1 \leq t_1 \leq T$ and $1 \leq t_2 \leq T$.  This
requires a running time of $O(N+T^2)$.


\section{Algorithm \avgdistbuild}

\begin{figure}[!ht]
	\begin{center}
		\begin{framed}
			\begin{tabbing}
				Algorithm \textbf{AvgDist-Build} \\
				\textbf{Input:} Node $t$ \\
				\textbf{Output:} Object list $B$ \\
				1. $L$ := list of objects in $t$ \\
				2. $B$ := Modify($L$) \\
				3. $c$ := number of children of $t$ \\
				4. \textbf{for} \= $i$ = $1$ to $c$ \\
				5. \> $CB[i]$ := AvgDist-Build($t.child[i]$) \\
				6. \> \textbf{for} \= \textbf{each} $co \in CB[i]$\\
				7. \> \> \textbf{if} \= $\exists o$ := Find($co.id$, $B$) \\
				8. \> \> \> $o.dist$ := \= $o.dist + co.dist$ \\
				   \> \> \> \> $+ co.count \times t.edge[i]$ \\
				9. \> \> \> $o.count$ := $o.count + co.count$ \\
				10. \> \> \textbf{end if} \\
				11. \> \textbf{end for} \\
				12. \textbf{end for}
			\end{tabbing}
		\end{framed}
		\vspace*{-4mm}
		\caption{The \emph{Build} phase of the \avgdist{} algorithm.}
		\label{fig:algoavgbuild}
	\end{center}
	\vspace*{-6mm}
\end{figure}

\end{document}